\documentclass[runningheads]{comsoc2020}
\usepackage{amsmath}
\usepackage{array,xspace,multirow,hhline,graphicx,xcolor,tikz,colortbl,tabularx,amsmath,amssymb,amsfonts}
\usetikzlibrary{shapes}
\usetikzlibrary{arrows}
\usepackage{amsthm}
\usepackage{times}

\newcommand{\citet}[1]{\citeauthor{#1}~\shortcite{#1}}
\newcommand{\citep}{\cite}
\usepackage{graphicx}
\usepackage{caption}
\usepackage{natbib}
\usepackage{booktabs}
\usepackage{xcolor}
\usepackage{wrapfig}
\usepackage{amsthm}
\newcommand\eat[1]{}
\usepackage{pgfplots}
\usepackage{bbm}
\usepackage{url}

\usepackage{verbatim,ifthen}
\usepackage{pifont}
\usepackage{ifthen}
\usepackage{pgflibraryshapes}

\usepackage{nicefrac}
\usetikzlibrary{arrows}
\usepackage{longtable}

	\usepackage{varioref}

\tikzset{
  jumpdot/.style={mark=*,solid},
  excl/.append style={jumpdot,fill=white},
  incl/.append style={jumpdot,fill=black},
  rexcl/.append style={jumpdot,color=red,fill=white},
  rincl/.append style={jumpdot,fill=black,color=red},
}
\usepackage[ruled]{algorithm2e} % For algorithms

\SetAlFnt{\small}
\SetAlCapFnt{\small}
\SetAlCapNameFnt{\small}
\SetAlCapHSkip{0pt}
\IncMargin{-\parindent}
\theoremstyle{definition}
\newtheorem{example}{Example}

\DeclareMathOperator*{\argmax}{arg\,max}

\usepackage{mathtools}

\DeclarePairedDelimiter{\floor}{\lfloor}{\rfloor}

\definecolor{gray(x11gray)}{rgb}{0.75, 0.75, 0.75}

\theoremstyle{definition}
\newtheorem{definition}{Definition}
\newtheorem{proposition}{Proposition}
\newtheorem{corollary}{Corollary}
\newtheorem{theorem}{Theorem}
\newtheorem{lemma}{Lemma}
\newtheorem{innercustomthm}{Theorem}
\newenvironment{customthm}[1]
  {\renewcommand\theinnercustomthm{#1}\innercustomthm}
  {\endinnercustomthm}
\newtheorem{innercustomlem}{Lemma}
\newenvironment{customlem}[1]
  {\renewcommand\theinnercustomlem{#1}\innercustomlem}
  {\endinnercustomlem}
  \title{Nash Welfare and Facility Location}
  \author{Alexander Lam and Haris Aziz and Toby Walsh}
\begin{document}
%
% \title{Nash Welfare in the Facility Location Problem}

%\footnote{The corresponding author for the paper is Alexander Lam}
%
%\titlerunning{Abbreviated paper title}
% If the paper title is too long for the running head, you can set
% an abbreviated paper title here
%
% \author{Haris and Alex}
% \author{Haris Aziz\inst{1} \and
% Alexander Lam\inst{1}
% }
% %
% \authorrunning{H. Aziz and A. Lam}
% First names are abbreviated in the running head.
% If there are more than two authors, 'et al.' is used.
%
%\institute{UNSW Sydney, Sydney, Australia \and Data61 CSIRO, Sydney, Australia}
%
%\maketitle              % typeset the header of the contribution
%
\begin{abstract}
We consider the problem of locating a facility to serve a set of agents located along a line. The Nash welfare objective function, defined as the product of the agents' utilities, is known to provide a compromise between fairness and efficiency in resource allocation problems. We apply this welfare notion to the facility location problem, converting individual costs to utilities and analyzing the facility placement that maximizes the Nash welfare. We give a polynomial-time approximation algorithm to compute this facility location, and prove results suggesting that it achieves a good balance of fairness and efficiency. Finally, we take a mechanism design perspective and propose a strategy-proof mechanism with a bounded approximation ratio for Nash welfare.
\end{abstract}
\section{Introduction}
In the facility location problem, we are tasked with placing a facility in an optimal location along a line to serve a set of agents. Each agent incurs a cost proportional to their distance from the facility, and we can also interpret this as a case of single-peaked preferences. The facility location problem generalizes to many real-life problems. Some geographical examples include the placement of a public facility along a long street \citep{Miya01a}, or the placement of a wastewater plant along a river. Another example could have an agent's location along the line indicating their position on the political spectrum, and the facility placement could be the choice of representative that most optimally reflects the opinions of the agents \citep{FFG16a}. This problem can also be extended to graphs, and applied to finding an optimal router placement on a network \citep{Gao12a}.
% \haris{Give a few representative references in the intro on FL including a couple from AI. }
%\alex{added extension to graphs and placing router on a network. Should I add another example?}
% \haris{The politician example is not very realistic. Politicians place their positions based on their history, profile, political climate and to win power. } \alex{changed politician to representative. is this enough?}

In our approach to this problem, we convert the individual costs to utilities, and analyse the facility location placement which maximizes the product of these utilities, also known as the Nash welfare. We remark that the use of utilities allows us to quantify how little or how much each agent benefits from the facility placement. 
%\alex{Moved the Moulin reference to Related work for now, but isn't it good to justify this idea here by citing Moulin?} 
Denoting the mechanism that outputs the facility location maximizing Nash welfare as \textsc{NashFL}, we aim to address questions relating to its computation and its approximation of fairness and efficiency objectives. We also investigate the 
trade-off between strategy-proofness and maximizing Nash welfare.
%use of a strategy-proof mechanism to approximate the Nash welfare. 
%\haris{The previous sentence is more appropriate in the related work. Just say you look into Nash FL in detail. } \alex{What do you think of the new paragraph?}

We first approach this idea from a computational perspective. We find that the \textsc{NashFL} output can take an irrational form and hence cannot be represented by the rationals. Hence we explore the approximation of the optimal Nash solution. 
%\haris{Please explain the previous part in a better way. Just be direct. That you explore approximation because you show that the outcome can be irrational. } \alex{Addressed. Pls check new text.}

In many areas of social choice, the solution maximizing fairness often sacrifices efficiency, and vice versa. Consequently, researchers have turned to the solution maximizing the Nash welfare, often considered to be an intermediary objective function providing a compromise between fairness and efficiency measures. 
%When fairly allocating indivisible goods, the maximum Nash welfare solution is Pareto optimal and meets certain fairness constraints \citep{CKM+16a}. In the field of participatory budgeting, the solution maximizing Nash welfare meets Pareto optimality and Unanimous Fair Share \citep{ABM19a}, the latter property implying that a coalition of $k$ agents with identical preferences are guaranteed at least $\frac{k}{n}$ utility.
By analyzing how well the optimal Nash solution approximates egalitarian and utilitarian objectives in the worst case, we determine whether the fairness-efficiency tradeoff is retained in the facility location problem. We also investigate which fairness axioms are met by the Nash solution.

Finally, we look at the strategic aspects of the problem, considering the context where the agents' locations are private information, and that they may strategically misreport their location to attain a higher utility. 
% \haris{Mechanism design approach is about designing strategyproof mechanisms. Just say you look at strategic aspects. }\alex{Addressed. Pls check new text.}
Strategy-proof mechanisms, which make it optimal for agents to report truthfully, are used to discourage such strategic behaviour. We investigate how well the Nash welfare can be approximated by a strategy-proof mechanism, and examine the properties of one such mechanism with a bounded approximation ratio.

\paragraph{Contributions}

% \haris{Again the previous sentence sounds appropriate in related work. }
% \haris{Please show some energy in this section by pointing out some interesting aspects. It seems like a sentence by sentence list of propositions.   }\alex{Addressed. Pls check new text.}
In this work, we are the first to provide an analysis of the facility location maximizing the Nash welfare. By proving that the exact facility location cannot always be represented by the rationals, we show that it suffices to use an approximation algorithm, and subsequently give a polynomial-time algorithm that computes this facility location within a specified additive error. We then show that unlike the Midpoint or the Median solution, the Nash solution satisfies Unanimous Fair Share, a fairness property that guarantees a proportional amount of utility for each coalition of agents at the same location.
We also prove approximation ratio guarantees of egalitarian, utilitarian and Nash objectives by certain facility location mechanisms; Table 1 summarizes these findings. Lastly, we restrict to strategy-proof mechanisms and prove a somewhat negative result: no strategy-proof mechanism can provide a constant factor approximation of the optimal Nash welfare. Nevertheless, we also give a strategy-proof mechanism with a bounded approximation ratio for Nash welfare. Due to space restrictions, theorems and lemmas lacking proofs are proved in the appendix.
\begin{table}[h!]
	 \begin{center}
	 \scalebox{1}{
 \begin{tabular}{l  c c c c}
 \toprule
& OptEgal & OptUtil & OptNash 
\\
\midrule
Midpoint&$1$ & $2-\frac{2}{n}$ & $O(\mathbf{2^n})$\\
Median$^*$&$\infty$ &$1$ & $\infty$ \\
NashFL&$\frac{n}{2}$ &$\mathbf{[1.2,2]}$ & $1$\\
MidOrNearest$^*$ & $\frac{3}{2}$ & $2-\frac{2}{n}$ & $2^{n-2}$\\
  \bottomrule
 \end{tabular}
}
 \end{center}
   \caption{Worst-case approximation ratio guarantees of social welfare properties by specific solutions. $^*$ indicates that the mechanism is strategy-proof. All approximation ratio bounds are tight except those in bold.}
   \label{table:summary}
 \end{table}
 \vspace{-0.2cm}
\section{Related Work}

% \haris{Please start off with classical FL papers. You are hardly doing mechanism design so not sure why so much focus on SP and FL in the related work. Please go beyond the paper by Ariel and mention the most classic papers on FL. }

% \haris{Indicate how established the field is by citing some books on FL. }\alex{Addressed. Pls check new text.}

The facility location problem has been widely researched in mathematics for centuries, appearing as early as 1643 as the Fermat-Torricelli problem \citep{Ferm91a}. In recent decades, it has been applied to operations research from an optimization perspective, with a focus on minimizing transport costs. An overview of approaches and results can be found in \citep{Hekm09a} and \citep{MNS09a}. Many variations of the problem have been studied in the literature, such as the analysis of capacity-constrained facilities \citep{WZZ06a} and the consideration of distant agents \citep{CKMN01a}. A review of research surrounding facility location models when there is uncertainty is given by \citet{Snyd06a}. When dealing with agents that arrive in an online fashion, randomized algorithms can be used to maintain a set of facilities \citep{Meye01a}. Many variants of the facility location problem are known to be NP-Hard, and hence there has been much research on approximation algorithms, such as those proposed in \citep{STA97a}.
These approximations have gradually improved over the years, and been approached in differing angles \citep{ChWi99a, GuKh99a, ChGu05a}. 

In recent years, there has been much research on facility location mechanisms satisfying strategy-proofness, building off the analyses of single-peaked preferences by papers such as \cite{Mou80a, BoJo83a}. Research specific to strategy-proof facility location mechanisms has been initiated by \citet{PrTe13a}, where they give approximation ratio bounds for deterministic and randomized mechanisms, under the key objectives of social cost and maximum cost. Many variations of this area of study have been considered, such as the algorithmic and mechanism design approaches to capacitated facilities \citep{ACLP20a, ACL+20a}, the inclusion of externalities \citep{LMX+19a} and the analysis of weighted agents \citep{ZhLi14a}.
The facility location game has also been applied to activity scheduling \citep{XLD20a}, in which the activities are denoted by facilities taking up a bounded interval on a timeline.

%such as
The idea of converting costs to utilities in the facility location problem was proposed by \citet{Mou03a}, and it has since been used in various forms. While Moulin defines agent utilities as $1-cost$, which has also been used in \citep{ACLP20a}, \citep{MLYZ19a} scale an agent's `satisfaction' to 0 when the facility is as far as possible from the agent, and 1 when the facility is at the agent's location. While both settings restrict the interval to $[0,1]$, an agent at location $\frac{1}{2}$ with the facility placed at $0$ has $\frac{1}{2}$ utility under our definition yet $0$ satisfaction. Their paper also discusses the obnoxious facility location game, introduced by \citet{CYZ11a}, in which agents desire to be as far away from the facility as possible. The authors later extended their analysis from paths to networks, and formally defined utility in this setting as an agent's distance from the nearest facility \citep{CYZ13a}.

Our work is not the first to examine objectives differing from social cost and maximum cost: the least squares objective in both line and tree networks is discussed in \citep{FeWi13a}. Contrasting with this convex objective function, \citep{FoTz16a} presents a group strategy-proof, randomized mechanism for multiple facilities which achieves a bounded approximation ratio for any concave cost function. The Nash welfare objective differs from these cases as it is neither concave nor convex.

The idea of maximizing the product of utilities originates in the analysis of the bargaining problem in \citep{Nash50a}. The Nash welfare objective function has since been used in several areas of the social choice literature to find a reasonable compromise between fairness and efficiency. In the fair allocation of divisible goods, the allocations maximizing Nash social welfare coincide with those resulting from the competitive equilibrium from equal incomes (CEEI) solution, and therefore satisfy envy-freeness and Pareto optimality \citep{ArIn93a}.
When allocating indivisible goods, the Nash solution retains its Pareto optimality and achieves envy freeness up to one good \citep{CKM+16a}. In \citep{ABM19a}, a variation of the participatory budgeting model is discussed, and it is found that the optimal Nash solution is not only ex-ante efficient, but also satisfies certain fairness guarantees.
The Nash welfare has also been proposed as an objective for the facility location model of ambulance placement \citep{JaMa20a}, though in this work, the Nash welfare is a function of ambulance durations in certain configurations rather than distances.

\section{Preliminaries}
In our setting, we  have a set $N=\{1,\dots,n\}$ of $n$ agents. Each agent $i\in N$ is at location $x_i\in [0,1]$. We define the agent location profile as the vector $\mathbf{x}=(x_1,\dots,x_n)$, and we assume the locations are ordered such that $x_1\leq \dots \leq x_n$.
The agents are served by a single facility. This facility is placed by a deterministic mechanism, which is a function $f: [0,1]^n\rightarrow [0,1]$ that takes the location profile $\mathbf{x}$ and outputs a location for the facility. Under facility location $y$, agent $i\in N$ incurs cost $c(y,x_i)=|y-x_i|$ and has utility $u(y,x_i):=1-c(y,x_i)$.

%We define the (total) social cost of a mechanism as the sum of each agents' costs: $SC(f,\mathbf{x}):=\sum_i c(f(\mathbf{x}),x_i)$, and the maximum cost of a mechanism as the maximum cost incurred by an agent: $MC(f,\mathbf{x}):=\max_i c(f(\mathbf{x}),x_i)$.

Given an agent location profile $\mathbf{x}$ and facility location $y$, the utilitarian social welfare of a mechanism is the sum of the agents' utilities: $USW(y,\mathbf{x}):=\sum_i u(y,x_i)$, and the egalitarian social welfare of a mechanism is the minimum utility achieved by an agent: $ESW(y,\mathbf{x}):=\min_i u(y,x_i)$.

Finally, we define the Nash welfare of a mechanism as the product of the agent utilities: $Nash(y,\mathbf{x}):=\prod_{i}u(y,x_i)$. In this paper, we are primarily interested the mechanism which places the facility such that Nash welfare is maximized. We will define this mechanism as $\textsc{NashFL}(\mathbf{x}):=\argmax_{y\in[0,1]}\prod_{i}u(y,x_i)$.
\section{The Nash Solution: Structural Properties and Computation}
In \citep{Mou03a}, Moulin discusses locating a facility along a line so as to maximize Nash welfare. He observes that ``this optimum is neither easy to compute nor to interpret.'' We demonstrate that the optimum can take an irrational form.

For $2$ agents, the facility location optimizing Nash welfare is simply the midpoint of the two agents, so rational agent locations imply a rational facility location. However, for $3$ agents, the Nash welfare becomes a cubic polynomial. The derivative is therefore a quadratic polynomial, so it is intuitive that an irrational facility location can arise from rational agent locations. Below, we find an exact, analytical solution for the $\textsc{OptNash}$ output when there are exactly $3$ agents.

\begin{lemma}\label{lemma:irrational}
Suppose there are 3 agents at locations $x_1$, $x_2$ and $x_3$. Let $c=1-x_2^2+x_1x_2+x_2x_3-x_1x_3$. If $2x_1-2x_2+c\geq 0$ and $2x_2-2x_3+c\geq 0$, then \textsc{NashFL} places the facility at agent $x_2$. If $2x_1-2x_2+c\geq 0$ and $2x_2-2x_3+c<0$, then \textsc{NashFL} places the facility at location
\[\frac{(1+\alpha)- \sqrt{(1+\alpha)^2-3(2x_3-\beta)}}{3},\]
whilst if $2x_1-2x_2+c<0$ and $2x_2-2x_3+c\geq 0$, then \textsc{NashFL} places the facility at location
\[\frac{(-1+\alpha)+ \sqrt{(-1+\alpha)^2+3(2x_1+\beta)}}{3},\]
where $\alpha=x_1+x_2+x_3$ and $\beta=1-x_1x_2-x_2x_3-x_1x_3$.
\end{lemma}
\begin{theorem}
There exists a profile of rational agent locations such that \textsc{NashFL} places a facility at an irrational location.
\end{theorem}
\begin{proof}
Suppose we have 3 agents at locations $x_1=\frac17$, $x_2=\frac27$ and $x_3=\frac67$. By substituting these values into Lemma \ref{lemma:irrational}, we find that the optimal facility location maximizing Nash welfare is at $\frac{16- \sqrt{91}}{21}$, which is irrational. 
\end{proof}
A consequence of this result is that the exact solution cannot always be represented by the rationals, so when computing the \textsc{NashFL} solution, it can suffice to use an approximation algorithm. Next, we give an algorithm that computes an approximate solution in polynomial time.

	\begin{theorem}\label{comp}
		The solution with the maximum Nash welfare can be computed within additive error of $\epsilon>0$ in time that is polynomial in the input and $1/\epsilon$.
		\end{theorem}
		\begin{proof}
			We consider at most $n-1$ different cases corresponding to the segments $[x_j,x_{j+1}]$ between two reported consecutive agent locations. Within this interval, we can approximately compute the point that maximizes the Nash social welfare as follows. The utility of each agent $i$ is $u_i=1-(x_i-y)$ or $1-(y-x_i)$ depending on whether $x_i$ is right of $x_{j+1}$ or left of $x_j$. Therefore, each agent utility can be captured by linear inequalities. The objective is $\max \prod_{i\in N}u_i=\min -\sum_{i\in N}\log u_i$. 
			
If we restrict the facility's location to the interval $[x_j,x_{j+1}]$, the optimal location in the interval is the solution to the following program.
			\begin{align*}
				&\min -\sum_{i\in N}\log u_i&\\
				&u_i=1-(x_i-y) & \text{ if } x_i\geq x_{j+1}\\
				&u_i=1-(y-x_i) & \text{ if } x_i\leq x_{j}\\
				&y\geq x_{j}&\\
				&y\leq x_{j+1}&
				\end{align*}
We need to solve the above program for each of at most $n-1$ intervals:
\[[x_1,x_2],\ldots, [x_{n-1}, x_n].\] 
Next, we show that for one interval, the optimization can be done almost optimally in time that is polynomial in the input and $1/\epsilon$ where $\epsilon>0$ is the additive error. 		
			
In page 899 of \citep{Vazi12a} a program (1) is defined which by substituting functions $f_i$'s we get the program we need to solve for a particular segment.
The program can be approximately solved if there exists a separation oracle and if the program is indeed feasible. The program is feasible for every interval we consider as the interval does not consist of one point so there exists a point in the interval in which the utility of each agent is strictly positive. Also note that the separation oracle in our case is simply testing the linear constraint in the program which can be easily checked. Hence, the theorem follows.
			\end{proof}

%\section{General Properties of the Nash Solution}

Next, we give some general results regarding the Nash welfare as a function of the facility location. We first note that the Nash welfare is neither a concave nor a convex function, distinguishing it from previous work on concave/convex objective functions. Despite this, the Nash welfare is single-peaked as a function of the facility location. In other words, there is a unique facility placement that maximizes the Nash welfare, and the Nash welfare decreases as the facility location moves away from this optimum. 

\begin{theorem}\label{peak}
The Nash welfare as a function of the facility location is single-peaked.
\end{theorem}
We now show that this Nash welfare optimum is location-invariant, meaning that if each agent's location shifts by the same distance in a certain direction, the \textsc{NashFL} output also shifts by that distance in the same direction.

\begin{lemma}\label{invar}
\textsc{NashFL} is location-invariant.
%If each agent location is subtracted by the same constant $0 \leq c \leq x_1$, then the \textsc{NashFL} output is also subtracted by that constant. In other words, if \textsc{NashFL} solution for location profile $(x_1,\dots,x_n)$ is $y_{max}$, then for $0 \leq c \leq x_1$, the \textsc{NashFL} output for location profile $(x_1-c,\dots,x_n-c)$ is $y_{max}-c$.
\end{lemma}
This result allows us to simplify our proofs by setting $x_1=0$ without loss of generality. Our next result in this section shows that if an agent's location is shifted in one direction, the \textsc{NashFL} output does not shift in the other direction: it either remains in the same location or shifts in the same direction as the agent. 

\begin{lemma}\label{nonneg}
Suppose we have an agent location profile $\mathbf{x}=(x_1,\dots,x_n).$ If an agent's location $x_i$ is shifted left by some $c\in (0,x_i]$, then under the new agent location profile $\mathbf{x}'=(x_1,\dots,x_i-c,\dots,x_n),$ $\textsc{NashFL}(\mathbf{x}')\leq \textsc{NashFL}(\mathbf{x})$.
\end{lemma}
We also note that by symmetry, if an agent's location is shifted to the right, the optimal Nash facility location does not shift to the left. Using this result, we can prove that if a subset of the agents change locations, the facility location does not shift more than the agent with the greatest change in location.

\begin{lemma}\label{smaller}
Suppose we have two different agent location profiles $\mathbf{x}=(x_1,\dots,x_n)$ and $\mathbf{x}'=(x_1',\dots,x_n')$. The following inequality holds:
\[|\textsc{NashFL}(\mathbf{x})-\textsc{NashFL}(\mathbf{x'})|\leq \max_{i\in N} |x_i-x_i'|.\]
\end{lemma}

Lastly, we find the exact analytical solution for \textsc{NashFL} in the restricted case where agents can only take two distinct locations.
\begin{lemma}\label{yesyes}
Let there be $k$ agents at location $x$ (where $0<x\leq 1$) and $n-k$ agents at location $0$. If $\frac{k}{n}\geq \frac{1}{2-x}$, then \textsc{NashFL} places the facility at $x$. If $\frac{k}{n}\leq \frac{1-x}{2-x}$, then \textsc{NashFL} places the facility at $0$. If neither of these inequalities hold, then \textsc{NashFL} places the facility at $x-1+\frac{2k-kx}{n}$.
\end{lemma}

\begin{corollary}
For all $x\in (0,1)$, there exists some $n$ and $k\in \{1,\dots,n-1\}$ such that \textsc{NashFL} places the facility at  location $0$ or location $x$.
\end{corollary}
%\begin{proof}
%Let $x=1-c$ for some $c\in(0,1)$. Substituting this into inequality $k(1-y)-(n-k)(1-x+y)>0$ and rearranging the expression, we have $k>\frac{n}{1+c}$ which can hold for all $c\in (0,1)$ by having arbitrarily large $n$ and $k<n$. When this inequality holds, the optimal Nash facility location is at location $x$. A similar argument can be made for location $0$. \qed
%\end{proof}

\section{Approximation of Welfare Measures}

% \haris{This section seems to be going off in various tangents. Since the paper topic is NashFL,  Please focus on Nash and how well it apprimates various measures. Later discuss other rules. So organize the section by focusing on Nash's performance rather than divising according to objects. In a later subsection, you can discuss MED, MEDIAN, ETC. } \alex{Made a very slight restructuring and rewording. I think it's fine now. I need to define Mid and MED up here to introduce approximation ratio. The first two subsections start and focus with showing how well NashFL approximates egal/util and end by briefly comparing it with Med/Mid. The third subsection is somewhat different but properly introduced.}

In this section, we primarily examine the worst case ratio between the optimal welfare value and the welfare value resulting from the \textsc{NashFL} facility placement. We then make a further comparison by examining how well other mechanisms approximate the Nash welfare.
We first define the following mechanisms:
\begin{itemize}
\item $\textsc{Mid}$ is the midpoint mechanism which maximizes egalitarian social welfare,
\item $\textsc{Med}$ is the median mechanism which maximizes utilitarian social welfare.
\end{itemize}
Specifically, $\textsc{Mid}(\mathbf{x})=\frac{x_1+x_n}{2}$. If there are an even number of agents, $\textsc{Med}$ places the facility at the leftmost point of the optimal interval.
\begin{definition}
For egalitarian, utilitarian and Nash social welfare, we define the approximation ratio as the maximum ratio between the optimal welfare and the welfare from the facility location, over all possible agent location profiles.
\[\max_{\mathbf{x}\in [0,1]^n} \frac{ESW(\textsc{Mid}(\mathbf{x}),\mathbf{x})}{ESW(f(\mathbf{x}),\mathbf{x})}, \max_{\mathbf{x}\in [0,1]^n} \frac{USW(\textsc{Med}(\mathbf{x}),\mathbf{x})}{USW(f(\mathbf{x}),\mathbf{x})},\]
\[\max_{\mathbf{x}\in [0,1]^n} \frac{Nash(\textsc{NashFL}(\mathbf{x}),\mathbf{x})}{Nash(f(\mathbf{x}),\mathbf{x})}.\]
\end{definition}

\subsection{Egalitarian Social Welfare}

The egalitarian social welfare, or minimum utility attained by an agent, is a leximin measure of fairness analogous to the agents' maximum cost which is commonly used in the literature. We prove that $\textsc{NashFL}$ achieves a linear approximation ratio for egalitarian social welfare, and then make a comparison with the \textsc{Med} mechanism.
\begin{theorem}\label{egalal}
\textsc{NashFL} $\frac{n}{2}$-approximates the egalitarian social welfare.
\end{theorem}
\begin{proof}
Due to the location invariance of both \textsc{NashFL} and $\textsc{Mid}$, it suffices to only consider agent location profiles where $x_1$ is at location $0$ and $x_n$ is at some $x\in (0,1]$. Under these location profiles, $\textsc{Mid}$ always places the facility at $\frac{x}{2}$, so the location profile satisfying 
\[\max_{\mathbf{x}\in [0,1]^n} \frac{ESW(\textsc{Mid}(\mathbf{x}),\mathbf{x})}{ESW(\textsc{NashFL}(\mathbf{x}),\mathbf{x})}\]
maximizes the distance between the optimal Nash facility location and $\frac{x}{2}$. From Lemma \ref{nonneg}, this is achieved by having $n-1$ agents at $0$ and $1$ agent at $x$, or vice versa. Due to symmetry, we simply consider the former location profile.

We now upper bound the approximation ratio for two cases of $x$.\\
\textbf{Case 1} $(x\leq\frac{n-2}{n-1})$:\\
From Lemma \ref{yesyes}, if $\frac{1}{n}\leq \frac{1-x}{2-x}$, then the optimal Nash facility location is at $0$. Rearranging this, we have $x\leq \frac{n-2}{n-1}$. An optimal Nash facility location of $0$ corresponds to an egalitarian social welfare of $1-x$, whilst a facility location of $\frac{x}{2}$ corresponds to the egalitarian social welfare of $\frac{2-x}{2}$. Dividing these terms, we have the approximation ratio of $\frac{2-x}{2(1-x)}$. This ratio increases as $x$ increases, so under the constraint of $x\leq \frac{n-2}{n-1}$, we substitute $x= \frac{n-2}{n-1}$ to attain the maximum ratio of $\frac{n}{2}$.\\
\textbf{Case 2} $(x>\frac{n-2}{n-1})$:\\
From Lemma \ref{yesyes}, the optimal Nash facility location in this case is $x-1+\frac{2-x}{n}$. This corresponds to an egalitarian social welfare of $1-[x-(x-1+\frac{2-x}{n})]=\frac{2-x}{n}$, whilst the optimal egalitarian social welfare is $\frac{2-x}{2}$. Dividing these terms, we have the ratio of $\frac{n}{2}$. By exhaustion of cases, we have shown that no agent location profile can lead to an approximation ratio greater than $\frac{n}{2}$.

The case analysis has also shown that there exists an agent location profile that leads to a ratio of $\frac{n}{2}$, implying that the approximation ratio is at least $\frac{n}{2}$. The approximation ratio is therefore exactly $\frac{n}{2}$.
\end{proof}
In contrast, the $\textsc{Med}$ mechanism has an unbounded approximation ratio for egalitarian social welfare, as it permits a case where at least one agent has 0 utility (2 agents at $0$, 1 agent at $1$). 
\subsection{Utilitarian Social Welfare}
The utilitarian social welfare, or total utility achieved by the agents, is a commonly-used measure of efficiency. We prove that $\textsc{NashFL}$ achieves a constant approximation ratio for utilitarian social welfare, and then make a comparison with the \textsc{Mid} mechanism.
\begin{lemma}\label{xDD}
\textsc{NashFL} has an approximation ratio of at least $\frac{\sqrt{2}+1}{2}\approx 1.2$ for utilitarian social welfare.
\end{lemma}
\begin{proof}
Suppose we have $n-k$ agents at $0$ and $k$ agents at $1$, and without loss of generality that $n-k\geq k$. The optimal median mechanism places the facility at $0$, resulting in a utilitarian social welfare of $n-k$. From Lemma \ref{yesyes}, \textsc{NashFL} places the facility at $\frac{k}{n}$, resulting in a utilitarian social welfare of $\frac{k^2+(n-k)^2}{n}$. The ratio between the utilitarian social welfare of the optimal solution and the \textsc{NashFL} solution in this restricted domain is
\begin{align*}
\frac{USW(\textsc{Med}(\mathbf{x}),\mathbf{x})}{USW(\textsc{NashFL}(\mathbf{x}),\mathbf{x})}&=\frac{n(n-k)}{k^2+(n-k)^2}\\
&=\frac{n^2-kn}{2k^2+n^2-2kn}\\
&=\frac{1-r}{2r^2-2r+1},
\end{align*}
where $r=\frac{k}{n}$. By taking derivatives, we note that this ratio is maximized when $r=\frac{2-\sqrt{2}}{2}$, taking a value of $\frac{\sqrt{2}+1}{2}$. 
\end{proof}
\begin{lemma}\label{guar}
\textsc{NashFL} guarantees a utilitarian social welfare of at least $\frac{n}{2}$.
\end{lemma}
\begin{proof}
To prove this lemma, we show that a series of transformations, each with a non-positive net gain to total utility, can be applied to any arbitrary location profile to construct the location profile $(\underbrace{0,\dots,0}_{\lfloor \frac{n}{2}\rfloor},\underbrace{1,\dots,1}_{\lceil \frac{n}{2}\rceil})$, which has at least $\frac{n}{2}$ total utility under the \textsc{NashFL} mechanism.

We first start with location profile $\mathbf{x_0}=(x_1,\dots,x_n)$. Let $k$ be the number of agents to the left of the facility\footnote{If an agent is at the same location as the facility, we will say it is to the left of the facility.}, and let $n-k$ be the number of agents to the right of the facility. Without loss of generality, suppose that $n-k\geq k$. The first transformation shifts the $k$ agents to the left of the facility to location $0$, resulting in location profile $\mathbf{x_1}=(\underbrace{0,\dots,0}_{k},x_{k+1},\dots,x_n)$. Let $\Delta y$ be the change in facility location as a result of this transformation. By Lemma \ref{nonneg}, we have $\Delta y \leq 0$. The net change in total utility is $-\sum_{i=1}^k x_i+\Delta y ((n-k)-k)\leq 0$, the first term representing the change in utility of agents $x_1,\dots,x_k$ from their movements, and the second term representing the change in utility of all the agents from the facility movement. Therefore this transformation results in non-positive net change in total utility.

In this step, we start with location profile $\mathbf{x_1}=(\underbrace{0,\dots,0}_{k},x_{k+1},\dots,x_n)$. If $k=\lfloor \frac{n}{2}\rfloor$, this step can be skipped. Suppose that $k<\lfloor \frac{n}{2}\rfloor$. We first prove that $\textsc{NashFL}(\mathbf{x_1})\geq \frac{x_{k+1}}{2}$. It is easy to deduce that under a location profile with $\frac{n}{2}$ agents at $0$ and $\frac{n}{2}$ agents at $x_{k+1}\in (0,1]$, \textsc{NashFL} places the facility at $\frac{x_{k+1}}{2}$. Now since $k<\lfloor \frac{n}{2}\rfloor$, we can transform this location profile to $\mathbf{x_1}$ without shifting any agents to the left. By Lemma \ref{nonneg}, we have $\textsc{NashFL}(\mathbf{x_1})\geq \frac{x_{k+1}}{2}$. We now transform $\mathbf{x_1}$ by shifting the agent at $x_{k+1}$ to $0$. Let $y=\textsc{NashFL}(\mathbf{x_1})$ and $\Delta y \leq 0$ be the change in facility location as a result of this transformation. The net change in total utility is $[(x_{k+1}-y)-(y-0)]+\Delta y ((n-k-1)-(k+1))\leq 0$. The first term is non-positive as $y\geq \frac{x_{k+1}}{2}$, and the second term is non-positive as $k+1\leq \lfloor \frac{n}{2}\rfloor$. We continue to iteratively shift the left-most agent locations of $x_{k+1},\dots,x_n$ to location $0$ until there are $\lfloor \frac{n}{2}\rfloor$ agents at $0$, forming the agent location profile $\mathbf{x_2}=(\underbrace{0,\dots,0}_{\lfloor \frac{n}{2}\rfloor},x_{\lceil \frac{n}{2}\rceil},\dots,x_n)$. The same argument can be applied to show that each of these transformations have non-positive change in total utility.

Finally, we transform agent location profile $\mathbf{x_2}=(\underbrace{0,\dots,0}_{\lfloor \frac{n}{2}\rfloor},x_{\lceil \frac{n}{2}\rceil},\dots,x_n)$ to the profile $(\underbrace{0,\dots,0}_{\lfloor \frac{n}{2}\rfloor},\underbrace{1,\dots,1}_{\lceil \frac{n}{2}\rceil})$ by shifting the agents at $x_{\lceil \frac{n}{2}\rceil},\dots,x_n$ to location $1$. Again, let $\Delta y$ be the change in facility location. By Lemma \ref{smaller}, we have $\Delta y \leq \max_{i\in \{\lceil \frac{n}{2}\rceil,\dots,n\}}|x_i-1|$. Hence the net change in total utility is $\sum_{i=\lceil \frac{n}{2}\rceil}^n(x_i-1)+\Delta y (\lceil \frac{n}{2}\rceil - \lfloor \frac{n}{2}\rfloor)\leq 0$.

Now if $n$ is even, $\textsc{NashFL}(\underbrace{0,\dots,0}_{\frac{n}{2}},\underbrace{1,\dots,1}_{\frac{n}{2}})=\frac{1}{2}$, resulting in $\frac{n}{2}$ total utility. If $n$ is odd, $\textsc{NashFL}(\underbrace{0,\dots,0}_{\frac{n-1}{2}},\underbrace{1,\dots,1}_{\frac{n+1}{2}})=\frac{n+1}{2n}$, resulting in $\frac{n^2+1}{2n}$ total utility. We have shown that a sequence of transformations with non-positive change in total utility can be applied to any agent location profile to construct a location profile with at least $\frac{n}{2}$ total utility. Therefore the \textsc{NashFL} solution guarantees a total utility of at least $\frac{n}{2}$. 
\end{proof}
\begin{theorem}\label{ez}
\textsc{NashFL} has an approximation ratio of at most $2$ for utilitarian social welfare.
\end{theorem}
\begin{proof}
Let $y$ be the solution of the \textsc{NashFL} mechanism, and $y_{MED}$ be the solution of the median mechanism. Also suppose that $x_1=0$ and $x_n=x$, where $x\in (0,1]$. The approximation ratio for utilitarian social welfare is
\[\max_{\mathbf{x}\in [0,1]^n} \frac{USW(\textsc{Med}(\mathbf{x}),\mathbf{x})}{USW(\textsc{NashFL}(\mathbf{x}),\mathbf{x})}=\max_{\mathbf{x}\in [0,1]^n}\frac{n-\sum^n_{i=1}|x_i-y_{MED}|}{n-\sum^n_{i=1}|x_i-y|}.\]
The utilitarian welfare corresponding to $y_{MED}$ is at most $n$,
and from Lemma \ref{guar}, we have $n-\sum^n_{i=1}|x_i-y| \geq \frac{n}{2}.$
Therefore, the approximation ratio is at most $2$.
\end{proof}
We now turn to the \textsc{Mid} mechanism, showing that it also attains a constant approximation ratio.
\begin{theorem}\label{asdf}
\textsc{Mid} has an approximation ratio for utilitarian social welfare of $2-\frac{2}{n}$.
\end{theorem}

This approximation ratio asymptotically matches the \textsc{NashFL} mechanism's upper bound proven in Theorem \ref{ez}, meaning that in the asymptotic case, \textsc{NashFL} approximates the utilitarian social welfare at least as well as \textsc{Mid}.

It may seem intuitive to prove Theorem \ref{ez} by first showing that the \textsc{NashFL} output lies between \textsc{MID} and \textsc{MED}. However, this is not always the case.
\begin{example}
The \textsc{NashFL} facility location does not always lie between the midpoint and the median locations. Consider the location profile with $k$ agents at $0$, $k$ agents at $0.5$ and $1$ agent at $1$. For $k=1, 2$ the \textsc{NashFL} output is $0.5$, but for $k=3$ the \textsc{NashFL} output is approximately $0.446$. In comparison, the midpoint and median facility location is $0.5$.
\end{example}
\subsection{Nash Welfare}
In the previous subsections, we have examined the approximation ratios of utilitarian and egalitarian objective functions for the \textsc{NashFL} mechanism. To attain a better insight as to how these objectives affect the Nash welfare, we find the approximation ratios that $\textsc{Med}$ and $\textsc{Mid}$ have for the Nash welfare. It is immediately clear that $\textsc{Med}$ has an unbounded approximation ratio for optimal Nash welfare, as it permits a case where at least one agent has 0 utility. We now show that the $\textsc{Mid}$ mechanism has an exponential approximation ratio for optimal Nash welfare.
\begin{lemma}\label{ugly}
$\textsc{Mid}$ has an approximation ratio for Nash welfare of at least $\frac{2^n}{n}\left(\frac{n-1}{n}\right)^{n-1}$.
\end{lemma}
\begin{proof}
Suppose there are $n-1$ agents at $0$ and $1$ agent at $1$. \textsc{NashFL} places a facility at $\frac{1}{n}$, resulting in a Nash welfare of $\frac{(n-1)^{n-1}}{n^n}$. \textsc{Mid} places a facility at $\frac{1}{2}$, resulting in a Nash welfare of $\frac{1}{2^n}$. Dividing these terms, we obtain the ratio $\frac{2^n}{n}\left(\frac{n-1}{n}\right)^{n-1}$. 
\end{proof}
\begin{theorem}
$\textsc{Mid}$ has an approximation ratio for Nash welfare of $O(2^n)$.
\end{theorem}
\begin{proof}
Let $y_{MID}$ be the facility location resulting from $\textsc{Mid}$ and $y$ be the \textsc{NashFL} facility location.
Since each agent can have at most $1$ utility, we have $Nash(y;\mathbf{x}) \leq 1$. 
Under $\textsc{Mid}$, each agent is guaranteed at least $\frac{1}{2}$ utility, so we have $Nash(y_{MID};\mathbf{x}) \geq \frac{1}{2^n}$. Dividing these terms gives the approximation ratio upper bound of $2^n$. Combining this with Lemma \ref{ugly}, we have the approximation ratio of $O(2^n)$. 
\end{proof}
\section{Fairness of the Nash Solution}
In this section, we examine some fairness properties. It would be unfair if a mechanism provided little or no utility to a subset of agents, so we may want to ensure that the facility placement guarantees a reasonable amount of utility to each agent. As a pathological example, if we have $k+1$ agents at $0$ and $k$ agents at $1$, then the \textsc{Med} mechanism places the facility at $0$, resulting in nearly half of the agents having $0$ utility and not benefiting from the facility at all.
We therefore introduce \emph{Individual Fair Share}, a fairness measure discussed in participatory budgeting problems \citep{ABM19a}.
%In that context, a mechanism satisfies Individual Fair Share if each agent is guaranteed at least $\frac{1}{n}$ of the total resource. In the facility location problem, each agent can have at most $1$ utility. Hence, we will define the property as follows.
\begin{definition}
A facility location mechanism satisfies \emph{Individual Fair Share} if each agent is guaranteed at least $\frac{1}{n}$ utility.
\end{definition}
%Note that in the facility location context, the purpose of Individual Fair Share is to guarantee each agent can benefit from the facility placement to a certain extent, whilst in fair allocation, the Individual Fair Share axiom ensures that each agent receives their `fair share' of the goods. 

Now consider the edge example where we have $n-1$ agents at $0$ and $1$ agent at $1$. If we apply the \textsc{Mid} mechanism and place the facility at $\frac{1}{2}$, then the agents at $0$ may be upset that their potential utility has been significantly affected by a single agent at a distant location. We may therefore want to use a mechanism that provides a proportional level of utility for each coalition of agents at the same location. %\alex{<-justification sentences added}
We hence define \emph{Unanimous Fair Share}, a fairness property used in the context of participatory budgeting \citep{ABM19a}. In that context, if $k$ agents have identical preferences, they should be guaranteed at least $\frac{k}{n}$ of the total utility. It is therefore a stronger notion than Individual Fair Share. We define the property similarly below.

\begin{definition}
A facility location mechanism $f$ satisfies \emph{Unanimous Fair Share} if for each location profile $\mathbf{x}$ and each subset of agents $S$ at the same location, $u(f(\mathbf{x}),x_i)\geq \frac{|S|}{n}$ for all $i\in S$.
\end{definition}
Since the median rule allows cases where an agent can have 0 utility, it does not satisfy Unanimous Fair Share, let alone Individual Fair Share.

In \citep{ABM19a}, it is proven that the Max Nash Product rule satisfies Unanimous Fair Share in the context of participatory budgeting. Below, we show that in the facility location problem, the \textsc{NashFL} mechanism satisfies Unanimous Fair Share.
\begin{theorem}\label{afs}
\textsc{NashFL} satisfies Unanimous Fair Share.
\end{theorem}
\begin{proof}
The case where all $n$ agents are at the same location is trivial, as \textsc{NashFL} places the facility at this location, resulting in each agent receiving a utility of $1$.

Suppose there are $n$ agents, and that $k\in \{1,\dots,n-1\}$ agents are at the same location $x\in [0,1]$. Let $S$ be the set of these $k$ agents. From Lemma~\ref{nonneg}, we know that \textsc{NashFL} is monotonic with respect to the agent locations. Therefore to minimize the utility of the agents at $x$, we maximize their distance from the facility by placing the remaining $n-k$ agents at either $0$ or $1$, whichever is furthest from $x$. Without loss of generality we consider the former case. Lemma~\ref{yesyes} gives the three subcases.

If $\frac{k}{n}\geq \frac{1}{2-x}$, then $\textsc{NashFL}$ places the facility at $x$, which gives $1$ utility to all agents in $S$.

If $\frac{k}{n}\leq \frac{1-x}{2-x}$, then $\textsc{NashFL}$ places the facility at $0$, resulting in each agent in $S$ receiving $1-x$ utility. By rearranging the equality, we have $1-x\geq (2-x)\frac{k}{n}$ and hence UFS holds.

If neither of those inequalities hold, $\textsc{NashFL}$ places the facility at $x-1+\frac{2k-kx}{n}$, resulting in each agent in $S$ receiving $1-[x-(x-1+\frac{k}{n}(2-x))]=(2-x)(\frac{k}{n})$ utility. Therefore UFS holds for all subcases.

%By location invariance, we note this location profile is equivalent to the case where the agents of $S$ are at $0$ and the remaining $n-k$ agents are at $1-x$. Again using the non-negative gradient property, we minimize the utility of the agents at $S$ with respect to $x$ by setting $x=0$.

%Therefore the location profile minimizing the utility of the agents in $S$ places the agents of $S$ at location $0$ and the other $n-k$ agents at location $1$, or vice versa. Due to symmetry, it suffices to only consider the former location profile. From Lemma \ref{yesyes}, we know that \textsc{NashFL} places the facility at $\frac{n-k}{n}$. Therefore, each agent in $S$ attains $\frac{k}{n}$ utility, meaning $U_S=\frac{k^2}{n}$. The inequality $\frac{1}{|S|}U_S\geq \frac{|S|}{n}$ is satisfied, as $\frac{1}{k}\frac{k^2}{n}=\frac{k}{n}$. 
\end{proof}

As previously explained, although the midpoint rule surpasses \textsc{NashFL} in terms of maximizing the minimum utility, we find that it fails to satisfy the notion of Unanimous Fair Share. We write the proof formally below.

\begin{proposition}
The midpoint rule does not satisfy Unanimous Fair Share.
\end{proposition}
\begin{proof}
Let $S$ be the set of $n-1$ agents at location $0$, and let there be $1$ agent at $1$. The midpoint rule places the facility at $\frac{1}{2}$, resulting in $\frac{1}{|S|}U_S=\frac{1}{2}$. The inequality $\frac{1}{|S|}U_S\geq \frac{|S|}{n}$ is therefore not satisfied for $n\geq 3$, as $\frac{|S|}{n}=\frac{n-1}{n}$.
\end{proof}
% \section{Strategy-proof Mechanisms}

\section{Strategic Aspects}

Most of our current results revolve around the \textsc{NashFL} mechanism which places a facility at the location maximizing Nash welfare. Although this mechanism achieves certain fairness and efficiency guarantees, it is not strategy-proof. A mechanism is \emph{strategy-proof} if no agent can increase their own utility by misreporting their location. Take the basic example where $x_1=0$ and $x_2=0.5$. Since $\textsc{NashFL}(x_1,x_2)=\frac{1}{2}(x_1+x_2)$, agent $2$ can misreport $x_2'=1$ to have the facility placed at her location.

Strategy-proof mechanisms are often employed in contexts where strategic behaviour can be problematic. For example, the \textsc{Med} mechanism is strategy-proof and optimal for utilitarian social welfare. However, as previously explained, it also has an unbounded approximation ratio for Nash Welfare. In fact, we find that the Nash welfare cannot be approximated up to a constant factor by any strategy-proof mechanism.

\begin{theorem}\label{sp}
No deterministic strategy-proof mechanism provides a constant factor approximation of the Nash welfare.
\end{theorem}
\begin{proof}
Suppose there exists a strategy-proof mechanism which provides a constant factor approximation of $\rho$ of the optimal Nash welfare for some $\rho\geq 1$, Consider $k$ agents at $0$ and $k$ at $\frac{1}{4}$. The optimal Nash welfare has the facility located at $\frac{1}{8}$ and is $(\frac{7}{8})^{2k}$. As the facility moves from $\frac{1}{8}$ towards $\frac{1}{4}$, the Nash welfare drops reaching $(\frac{3}{4})^k$ when the facility is at $\frac{1}{4}$. This compares to the optimal Nash welfare of $(\frac{7}{8})^{2k}$ or $(\frac{49}{64})^k$, Note that $\frac{49}{64}$ is strictly larger than $\frac{3}{4}$ so $\frac{49}{64}$ is strictly larger than $(\frac{3}{4})^k$. With the facility at $\frac{1}{4}$, the approximation ratio of the optimal Nash welfare is $(\frac{7}{8})^{2k}/(\frac{3}{4})^k$ or $(\frac{196}{192})^k$. For large enough $k$, this exceeds $\rho$. That is, there exists $k'$ such that for $k>k'$ we have $(\frac{196}{192})^k>\rho$ and the facility must be located strictly to the left of $\frac{1}{4}$. 

We now consider what happens when the $k$ agents at $\frac{1}{4}$ misreport their location as $x$ which varies smoothly from $x=\frac{1}{4}$ to $x=1$, The facility must remain to the left of $\frac{1}{4}$ as the mechanism is partially group strategy-proof and a group of agents jointly misreporting their location cannot result in a better outcome. A mechanism is partially group strategy-proof iff no group of agents at the same location can individually benefit if they misreport simultaneously. Any strategy proof mechanism (such as the one we have here by assumption) is also partially group strategy proof (Lemma 2.4 in \citep{FoTz14a}). When $k$ agents are at $0$ and $k$ at $1$, the optimal Nash welfare is $\frac{1}{2^{2k}}$ with the facility located at $\frac{1}{2}$. However, we have argued that with $k$ agents reporting location $0$ and $k$ reporting $1$, the facility must be located to the left of $\frac{1}{4}$. This gives a Nash welfare less than $\frac{3^k}{4^{2k}}$. The approximation ratio in this situation is then at least $\frac{1}{2^{2k}}$. This is, at least $(\frac{4}{3})^k$. Since $\frac{4}{3}>\frac{196}{192}>1$, we have $(\frac{4}{3})^k>(\frac{196}{192})^k>\rho$. That is, the mechanism fails to meet the approximation ratio of $\rho$. 
\end{proof}
We can, however, obtain a bounded approximation ratio for Nash welfare by a strategy-proof, Pareto Optimal and anonymous mechanism as shown below.

The \textsc{MidOrNearest} mechanism places the facility at $\frac12$ if $x_1\leq \frac12 \leq x_n$, else it places the facility at the agent closest to $\frac12$. This is also known as the \textsc{Moderate} mechanism as defined in \citep{DrLa19a}.
\begin{theorem}
\textsc{MidOrNearest} is strategyproof, Pareto optimal, anonymous and has an approximation ratio of $2^{n-2}$ for Nash welfare $(n\geq 3)$.
\end{theorem}
\begin{proof}
Note that \textsc{MidOrNearest} is equivalent to the phantom median mechanism which places the facility at $Median\{x_1,\dots,x_n,p_1,\dots,p_{n-1}\}$, where $p_1=\dots=p_{n-1}=\frac{1}{2}$. Corollary 2 of \citep{MaMo11a} states that phantom median mechanisms satisfy strategyproofness, Pareto optimality and anonymity. Hence, \textsc{MidOrNearest} is strategyproof, Pareto optimal and anonymous.
%To prove strategy-proofness, there are two cases to consider. In the first case, the facility is located at $\frac{1}{2}$ but a strategic misreport of one of the agents changes this. There are two subcases where this can happen. In the first, the misreporting agent is at $\frac{1}{2}$ or to the right, and all other agents are to the left of $\frac{1}{2}$. The agent can only move the facility if they now report a location to the left of $\frac{1}{2}$. But this can only hurt the outcome for them. The second subcase is symmetric. This leaves the second case in which the facility is located at the nearest agent to $\frac{1}{2}$. This agent has no incentive to misreport. Any other agent can also only move the facility further away. Hence, the mechanism is strategy-proof.

To determine the approximation ratio, there are four extreme cases to consider. For each mode of the mechanism, there is one extreme case (and its symmetry). In the first extreme case, $n-1$ agents are at $0$ and one is at $\frac{1}{2}$. Suppose the facility is located at $x$ with $0\leq x \leq \frac{1}{2}$. Then the Nash welfare is $(\frac{1}{2}+x)(1-x)^{n-1}$. For $n\geq 3$, the optimal Nash welfare of $\frac{1}{2}$ is with the facility located at $x=0$. On the other hand, the \textsc{MidOrNearest} mechanism locates the facility at $x=\frac{1}{2}$, giving a Nash welfare of $\frac{1}{2^n}$. The approximation ratio is therefore $\frac{2}{n}(\frac{2(n-1)}{n})^{n-1}$. This equals $2^{n-2}$ when $n=2$ and is smaller than $2^{n-2}$ when $n>2$. There is a symmetric extreme case with $n-1$ agents at $1$ and one at $0$. The worst case for the approximation ratio is then $2^{n-2}$. 
\end{proof}

Although this mechanism has an exponential approximation ratio for optimal Nash welfare, it has constant approximation ratio guarantees for both utilitarian and egalitarian social welfares. This may suggest that the optimal Nash welfare approximation ratio is more sensitive than our other measures.
%\begin{lemma}\label{po}
%Any Pareto optimal facility location mechanism has an approximation ratio of at most $n-1$ for utilitarian social welfare.
%\end{lemma}
\begin{theorem}\label{thm:midornear USW}
\textsc{MidOrNearest} is a $(2-\frac{2}{n})-$approximation of the utilitarian social welfare.
\end{theorem}

\begin{theorem}\label{midoregal}
\textsc{MidOrNearest} is a $\frac{3}{2}-$approximation of the egalitarian social welfare. No strategy-proof mechanism has a smaller approximation ratio.
\end{theorem}
%We note that this mechanism achieves a better egalitarian bound than \textsc{NashFL}, yet has an exponential approximation ratio for optimal Nash welfare. This is likely due to  of utilitarian social welfare, suggesting that a good approximation of the utilitarian objective is necessary to attaining a good approximation of the Nash objective.

\section{Discussion and Future Work}
In this paper, we have studied the Nash welfare objective in the facility location problem. When agent strategic behaviour is not a concern, the \textsc{NashFL} mechanism is a reasonably balanced option. It can be approximated up to a specified additive error in polynomial time, and it attains reasonable bounds and properties of fairness and efficiency. The Nash solution surpasses the median solution in terms of fairness, and is asymptotically at least as efficient as the midpoint solution in terms of its utilitarian social welfare approximation ratio. It also satisfies Unanimous Fair Share, a fairness property that even the midpoint solution does not satisfy. The results are more negative when we restrict to a strategy-proof mechanism domain: no strategy-proof mechanism can approximate the Nash welfare up to a constant factor. However, we propose the \textsc{MidOrNearest} mechanism, which is Pareto Optimal, anonymous and has a bounded approximation ratio for optimal Nash welfare. We also prove that this mechanism meets a linear approximation bound for utilitarian social welfare and a constant bound for egalitarian social welfare.

There are many extensions for this work. Some natural variations of the problem discussed by \citet{PrTe13a} are the scenarios with 2 facilities and/or randomized mechanisms. The 2-dimensional facility location problem with both Euclidean and Manhattan metrics could also be considered, such as by \citet{Wals20a} and \citet{GoHa20a}. We could also introduce capacity constraints for the setting with multiple facilities, in which each facility has a maximum number of agents it can serve. We note that for an unbounded number of capacitated facilities, computing a Nash welfare maximizing solution is NP-hard. This follows directly from the reduction by \citet{ACL+20a}, where it is shown that it is NP-complete to check whether these exists a solution in which each agent gets zero cost even when there is no spare capacity in the capacity-constrained facility location problem. It follows that computing a Nash welfare maximizing solution is NP-hard. Although \textsc{NashFL} is not strategy-proof, it may meet a weaker notion of strategy-proofness, and may be less manipulable than other non-strategy-proof mechanisms, such as the midpoint mechanism. The question remains whether a strategy-proof mechanism can provide a linear approximation of the optimal Nash welfare.  Finally, we would like to tighten the bound on the \textsc{NashFL} mechanism's approximation ratio for utilitarian social welfare.
%\bibliographystyle{named}

%\bibliography{abb,NashFL}

\begin{contact}
Alexander Lam\\
University of New South Wales\\
Sydney, Australia\\
\email{alexander.lam1@unsw.edu.au}
\end{contact}

\begin{contact}
Haris Aziz\\
University of New South Wales \\
Sydney, Australia\\
\email{haris.aziz@gmail.com}
\end{contact}
\begin{contact}
Toby Walsh\\
University of New South Wales \\
Sydney, Australia\\
\email{Toby.Walsh@data61.csiro.au}
\end{contact}
\newpage
\appendix
\section{Proof of Lemma \ref{lemma:irrational}}
\begin{customlem}{1}
Suppose there are 3 agents at locations $x_1$, $x_2$ and $x_3$. Let $c=1-x_2^2+x_1x_2+x_2x_3-x_1x_3$. If $2x_1-2x_2+c\geq 0$ and $2x_2-2x_3+c\geq 0$, then \textsc{NashFL} places the facility at agent $x_2$. If $2x_1-2x_2+c\geq 0$ and $2x_2-2x_3+c<0$, then \textsc{NashFL} places the facility at location
\[\frac{(1+\alpha)- \sqrt{(1+\alpha)^2-3(2x_3-\beta)}}{3},\]
whilst if $2x_1-2x_2+c<0$ and $2x_2-2x_3+c\geq 0$, then \textsc{NashFL} places the facility at location
\[\frac{(-1+\alpha)+ \sqrt{(-1+\alpha)^2+3(2x_1+\beta)}}{3},\]
where $\alpha=x_1+x_2+x_3$ and $\beta=1-x_1x_2-x_2x_3-x_1x_3$.
\end{customlem}
\begin{proof}
Let $\mathbf{x}=(x_1,x_2,x_3)$ and $y=\textsc{NashFL}(\mathbf{x})$. Due to Pareto optimality, we must have $x_1\leq y\leq x_3$. At facility location $y$, the Nash welfare is
\[Nash(y;\mathbf{x}) = (1-(y-x_1))(1-(x_3-y))(1-|y-x_2|).\]
To find an expression for $y$ in terms of the agent locations such that the Nash welfare is maximized, we observe the derivative of the Nash welfare with respect to $y$. We now consider the cases $y\geq x_2$ and $y\leq x_2$.\\
\textbf{Case 1:} ($y\geq x_2$)\\
To simplify our expression, we let $\alpha=1+x_1+x_2+x_3$ and $\beta = 2x_3-1+x_1x_2+x_2x_3+x_1x_3$.
Expanding the Nash Welfare expression, we have
\begin{align*}
Nash(y;\mathbf{x}) &= (1-(y-x_1))(1-(x_3-y))(1-(y-x_2))\\
&=y^3-\alpha y^2+\beta y+(1+x_1)(1+x_2)(1-x_3).
\end{align*}
The derivative of the Nash welfare with respect to $y$ is
\[\frac{dNash(y;\mathbf{x})}{dy}=3y^2-2\alpha y+\beta,\]
and it is equal to zero at
\[y=\frac{\alpha\pm \sqrt{\alpha^2-3\beta}}{3}.\]
Since the $y^3$ coefficient in the Nash cubic polynomial is positive, the local minimum point must lie to the right of the local maximum point. Therefore the Nash welfare attains its local maximum at
\[y_{max}=\frac{\alpha- \sqrt{\alpha^2-3\beta}}{3}.\]
Recall that we assume that $y\geq x_2$, so if $y_{max}< x_2$, then the Nash welfare is maximized at $y=x_2$, as it is strictly decreasing for $x_2<y\leq 1$.

Now
\begin{align*}
&y_{max}<x_2\\
&\iff \alpha^2-3\beta >(1+x_1-2x_2+x_3)\\
&\iff 2x_2-2x_3+1-x^2_2+x_1x_2+x_2x_3-x_1x_3>0,
\end{align*}
so if this inequality is satisfied, then the optimum Nash welfare in this case is at $y=x_2$.
We construct a similar argument for the other case:\\
\textbf{Case 2:} ($y\leq x_2$)

Again to simplify our expression, we let $\delta=-1+x_1+x_2+x_3$ and $\gamma=2x_1+1-x_1x_2-x_2x_3-x_1x_3$.
The Nash welfare becomes 
\begin{align*}
Nash(y;\mathbf{x}) &= (1-(y-x_1))(1-(x_3-y))(1-(x_2-y))\\
&=-y^3+\delta y^2+\gamma y+(1+x_1)(1-x_2)(1-x_3).
\end{align*}
The derivative of the Nash welfare with respect to $y$ is
\[\frac{dNash(y;\mathbf{x})}{dy}=-3y^2+2\delta y+\gamma,\]
and it is equal to zero at
\[y=\frac{\delta \pm \sqrt{\delta^2+3\gamma}}{3}.\]
Since the $y^3$ coefficient in the Nash cubic polynomial is negative, the local minimum point must lie left of the local maximum point. Therefore the Nash welfare attains its local maximum at
\[y_{max}=\frac{\delta + \sqrt{\delta^2+3\gamma}}{3}.\]
Recall that we assume that $y\leq x_2$, so if $y_{max}> x_2$, then the Nash welfare is maximized at $y=x_2$, as it is strictly increasing for $0\leq y < x_2$.

Now
\begin{align*}
&y_{max}>x_2\\
&\iff \delta^2+3\gamma>(1-x_1+2x_2-x_3)\\
&\iff 2x_1-2x_2+1-x^2_2+x_1x_2+x_2x_3-x_1x_3>0,
\end{align*}
so if this inequality is satisfied, then the optimum Nash welfare in this case is at $y=x_2$. The theorem statement follows.
\end{proof}
\section{Proof of Theorem \ref{peak}}
\begin{customthm}{3}
The Nash welfare as a function of the facility location is single-peaked.
\end{customthm}
\begin{proof}
By the Extreme Value Theorem, the Nash welfare must have a local maximum point on $[0,1]$. Denote this point as $y_{OPT}$. Since the Nash welfare is a piecewise non-constant polynomial, there are no neighbourhoods of points on $[0,1]$ such that $\frac{dNash(y;\mathbf{x})}{dy}=0$ for all points in the neighbourhood. Therefore, $y_{OPT}$ must be a strict local maximum point. Note that we must have $x_1\leq y_{OPT} \leq x_n$ due to the Pareto Optimality of \textsc{NashFL}.

Without loss of generality, suppose that $x_k \leq y_{OPT}<x_{k+1}$ for some fixed $k\in \{1,\dots,n-1\}$, or $y_OPT=x_k$ for $k=n$. We now proceed to prove that the Nash welfare is single-peaked by showing that $Nash(y;\mathbf{x})$ is strictly increasing in $[0,y_{OPT})$ and strictly decreasing in $(y_{OPT},1]$. Specifically, we show that $\frac{dNash(y;\mathbf{x})}{dy}>0$ for $y\in [0,y_{OPT})\backslash \{x_1,\dots,x_k\}$ and $\frac{dNash(y;\mathbf{x})}{dy}<0$ for $y\in (y_{OPT},1]\backslash \{x_{k+1},\dots,x_n\}$\footnote{The derivative may not exist at points $x_1,\dots,x_n$, but $Nash(y;\mathbf{x})$ is continuous and there are countably many points where the derivative does not exist.}. 

When $y\in (x_k,x_{k+1})$, our Nash welfare expression becomes
\begin{align*}
Nash(y;\mathbf{x})&=(1-|y-x_1|)\dots(1-|y-x_n|)\\
&= \prod_{i=1}^k(1-y+x_i) \prod_{i=k+1}^n(1+y-x_{i}).
\end{align*}

Since $\frac{d}{dx}\left[\prod_{i=1}^nf_i(x)\right]=\left(\prod^n_{i=1}f_i(x)\right)\left(\sum_{i=1}^n\frac{f'_i(x)}{f_i(x)}\right)$, the derivative of this function with respect to $y$ is
\begin{equation}\label{dev}
\frac{dNash(y;\mathbf{x})}{dy}=Nash(y;\mathbf{x})\left(-\sum^k_{i=1}\frac{1}{1-y+x_i}+\sum^n_{i=k+1}\frac{1}{1+y-x_{i}}\right).
\end{equation}
Since $y_{OPT}$ is a strict local maximum, there exists $\epsilon>0$ such that $\frac{dNash(y;\mathbf{x})}{dy}<0$ for $y\in (y_{OPT},y_{OPT}+\epsilon]$. The Nash welfare is non-negative, so in the region $(y_{OPT},y_{OPT}+\epsilon]$, the sum of fractions is negative. If $y_1>y_2$, $\frac{1}{1-y_1+x_i}>\frac{1}{1-y_2+x_i}$ for $i\in \{1,\dots,k\}$ and $\frac{1}{1+y_1-x_{i}}<\frac{1}{1+y_2-x_{i}}$ for $i\in \{k+1,\dots,n\}$. We therefore extend the previous interval to $\frac{dNash(y;\mathbf{x})}{dy}<0$ for $y\in (y_{OPT},x_{k+1})$.

Now consider $y\in (x_{k+1},x_{k+2})$. The Nash welfare expression changes to
\[Nash(y;\mathbf{x})= \prod_{i=1}^{k+1}(1-y+x_i)\prod_{i=k+2}^{n}(1+y-x_{i}),\]
so the derivative becomes
\[\frac{dNash(y;\mathbf{x})}{dy}=Nash(y;\mathbf{x})\left(-\sum_{i=1}^{k+1}\frac{1}{1-y+x_i}+\sum_{i=k+2}^{n}\frac{1}{1+y-x_{i}}\right).\]
This derivative is negative, as the $\frac{1}{1+y-x_{k+1}}$ term has changed to $-\frac{1}{1-y+x_{k+1}}$, and each individual fraction decreases as $y$ increases. By applying the same argument to the regions $(x_{k+2},x_{k+3}),\dots,(x_{n-1},x_n)$, we see that $\frac{dNash(y;\mathbf{x})}{dy}<0$ for $y\in (y_{OPT},x_n)\backslash \{x_{k+1},\dots,x_{n-1}\}$. Furthermore, $\frac{dNash(y;\mathbf{x})}{dy}<0$ for $y\in (x_n,1]$ as all of the fraction terms in the derivative become negative.

We now similarly show that $\frac{dNash(y;\mathbf{x})}{dy}>0$ for $y\in [0,y_{OPT})\backslash \{x_1,\dots,x_k\}$. If $y_{OPT}>x_k$, we first show that $\frac{dNash(y;\mathbf{x})}{dy}>0$ for $y\in (x_k,y_{OPT})$ (this is not necessary if $y_{OPT}=x_k$). $y_{OPT}$ is a strict local maximum, so there exists $\epsilon>0$ such that $\frac{dNash(y;\mathbf{x})}{dy}>0$ for $y\in [y_{OPT}-\epsilon,y_{OPT})$. The non-negativity of the Nash welfare implies that the sum of fractions in \ref{dev} is positive in this interval. Note that decreasing $y$ causes the sum of fractions to increase, so we have $\frac{dNash(y;\mathbf{x})}{dy}>0$ for $y\in (x_k,y_{OPT})$.

Now when we decrease $y$ to the interval $(x_{k-1},x_k)$, the derivative becomes
\[\frac{dNash(y;\mathbf{x})}{dy}=Nash(y;\mathbf{x})\left(-\sum_{i=1}^{k-1}\frac{1}{1-y+x_i}+\sum_{i=k}^{n}\frac{1}{1+y-x_{i}}\right),\]
which is positive as the $-\frac{1}{1-y+x_k}$ term changes to $\frac{1}{1+y-x_k}$, and the individual fractions increase as $y$ decreases. The same argument can be applied to regions $(x_{k-2},x_{k-1}),\dots,(x_1,x_2)$ to show that $\frac{dNash(y;\mathbf{x})}{dy}>0$ for $y\in (x_1,y_{OPT})\backslash \{x_1,\dots,x_k\}$. The derivative is also positive for $y\in [0,x_1)$ as all of the fraction terms become positive. Since $\frac{dNash(y;\mathbf{x})}{dy}>0$ for $y\in [0,y_{OPT})\backslash \{x_1,\dots,x_k\}$ and $\frac{dNash(y;\mathbf{x})}{dy}<0$ for $y\in (y_{OPT},1]\backslash \{x_{k+1},\dots,x_n\}$, the Nash welfare is strictly increasing for all $y\in [0,y_{OPT})$ and strictly decreasing for all $y\in (y_{OPT},1]$. We conclude that it is single-peaked.
\end{proof}
\section{Proof of Lemma \ref{invar}}
\begin{customlem}{2}
\textsc{NashFL} is location-invariant.
\end{customlem}
\begin{proof}
Suppose we have location profile $\mathbf{x}=(x_1,\dots,x_n)$. The Nash welfare expression for facility placement $y$ is
\[Nash(y;\mathbf{x})=(1-|y-x_1|)(1-|y-x_2|)\dots (1-|y-x_n|).\]
Let $\mathbf{x'}=(x_1+c,\dots,x_n+c)$ be the location profile where a constant $c\in [-x_1,1-x_n]$ has been added to each agent's location. The Nash welfare for this location profile is
\[Nash(y;\mathbf{x'})=\prod_{i=1}^n(1-|y-x_i-c|).\]
Denote $y_{max}=\textsc{NashFL}(\mathbf{x})$. The Nash welfare at facility location $y_{max}+c$ and agent location profile $\mathbf{x'}$ is the same as that of facility location $y_{max}$ in agent location profile $\mathbf{x}$:
\begin{align*}
Nash(y_{max}+c;\mathbf{x'})&=\prod_{i=1}^n(1-|y_{max}+c-x_i-c|)\\
&=Nash(y_{max};\mathbf{x}).
\end{align*}
We now prove by contradiction that no other facility location leads to a higher Nash welfare. Suppose that there exists some facility location $y'+c\neq y_{max}+c$ such that $Nash(y'+c;\mathbf{x'})>Nash(y_{max}+c;\mathbf{x'})$. We have the following inequality:
\begin{align*}
Nash(y';\mathbf{x})&=\prod_{i=1}^n(1-|y'+c-x_i-c|)\\
&=Nash(y'+c;\mathbf{x'})\\
&>Nash(y_{max}+c;\mathbf{x'})\\
&=Nash(y_{max};\mathbf{x}).
\end{align*}
This contradicts the assumption that $y_{max}$ maximizes the Nash welfare for $\mathbf{x}$. Therefore, $y_{max}+c$ is the optimal facility location for location profile $(x_1+c,\dots,x_n+c)$. 
\end{proof}

\section{Proof of Lemma \ref{nonneg}}
\begin{customlem}{3}
Suppose we have an agent location profile $\mathbf{x}=(x_1,\dots,x_n).$ If an agent's location $x_i$ is shifted left to $x_i'$ where $x_i'<x_i$, then under the new agent location profile $\mathbf{x}'=(x_1',\dots,x_i',\dots,x_n')=(x_1,\dots,x_i',\dots,x_n),$ $\textsc{NashFL}(\mathbf{x}')\leq \textsc{NashFL}(\mathbf{x})$.
\end{customlem}
\begin{proof}
Suppose that $x_k\leq \textsc{NashFL}(\mathbf{x})<x_{k+1}$ for some fixed $k\in \{1,\dots,n\}$ (we denote $x_{n+1}=1$)\footnote{We assume that $\textsc{NashFL}(\mathbf{x})<1$ and ignore the case where $\textsc{NashFL}(\mathbf{x})=1$, as this facility location cannot be shifted to the right.}. To prove this result, we show that the function $Nash(y;\mathbf{x}')$ is strictly decreasing w.r.t. $y$ in the interval $(\textsc{NashFL}(\mathbf{x}),1]$. Specifically, we prove that the derivative $\frac{dNash(y;\mathbf{x}')}{dy}<0$ for $y\in (\textsc{NashFL}(\mathbf{x}),1]\backslash \{x_{k+1}',\dots,x_n'\}$. This implies that any facility location in that interval cannot maximize the Nash welfare for agent location profile $\mathbf{x}'$, meaning that $\textsc{NashFL}(\mathbf{x}')\leq \textsc{NashFL}(\mathbf{x})$.

Due to the single-peaked nature of the Nash welfare (as proven in Theorem \ref{peak}), we have
\begin{align*}
\frac{dNash(y;\mathbf{x})}{dy} &= Nash(y;\mathbf{x})\bigg(-\sum_{j=1}^k \frac{1}{1-y+x_j}+\sum_{j=k+1}^n\frac{1}{1+y-x_j}\bigg)\\
&<0
\end{align*}
for $y\in (\textsc{NashFL}(\mathbf{x}),x_{k+1})$. Since the Nash welfare must be positive when $y\in [0,1]$, we deduce that the sum of fractions is negative for $y\in (\textsc{NashFL}(\mathbf{x}),x_{k+1})$.

%Specify that nash welfare is still single-peaked after the shift.
Now consider the agent location profile $\mathbf{x}'=(x_1',\dots,x_n')$ where $x_j'=x_j$ for $j\neq i$ and $x_i'<x_i$. By taking cases, we show that the derivative of the Nash welfare corresponding to this new location profile $\frac{dNash(y;\mathbf{x}')}{dy}$ is also negative for $y\in (\textsc{NashFL}(\mathbf{x}),1]\backslash \{x_{k+1}',\dots,x_n'\}$.

\textbf{Case 1:}

Suppose that after the shift, there are still $k$ agents left of $\textsc{NashFL}(\mathbf{x})$\footnote{Note $\textsc{NashFL}(\mathbf{x})$ is defined for $\mathbf{x}=(x_1,\dots,x_n)$ and does not change after the shift.}. If $x_i\leq \textsc{NashFL}(\mathbf{x})$, we have the following derivative expression
\begin{align*}
\frac{dNash(y;\mathbf{x}')}{dy}=&Nash(y;\mathbf{x}')\bigg(-\frac{1}{1-y+x_i'}-\sum_{\substack{j=1 \\ j\neq i}}^k\frac{1}{1-y+x_j}+\sum_{j=k+1}^n\frac{1}{1+y-x_j}\bigg),
\end{align*}
%=& Nash(y;\mathbf{x}')\bigg(-\frac{1}{1-y+x_1}-\dots-\frac{1}{1-y+x_j-c}-\dots\\
%&-\frac{1}{1-y+x_k}+\frac{1}{1+y-x_{k+1}}+\dots+\frac{1}{1+y-x_n}\bigg)\\
and if $x_i> \textsc{NashFL}(\mathbf{x})$, we have
\begin{align*}
\frac{dNash(y;\mathbf{x}')}{dy}=&Nash(y;\mathbf{x}')\bigg(-\sum_{j=1}^k\frac{1}{1-y+x_j}+\sum_{\substack{j=k+1 \\ j\neq i}}^n\frac{1}{1+y-x_j}+\frac{1}{1+y-x_i'}\bigg)
\end{align*}
for $y\in (\textsc{NashFL}(\mathbf{x}),x_{k+1}')$.
%& Nash(y;\mathbf{x}')\bigg(-\frac{1}{1-y+x_1}-\dots-\frac{1}{1-y+x_k}\\
%&+\frac{1}{1+y-x_{k+1}}+\dots+\frac{1}{1+y-x_j+c}+\dots+\frac{1}{1+y-x_n}\bigg)\\
Recalling that $-\sum_{j=1}^k \frac{1}{1-y+x_j} +\sum_{j=k+1}^n\frac{1}{1+y-x_j}<0$ for $y\in (\textsc{NashFL}(\mathbf{x}),x_{k+1})$ and noting that $-\frac{1}{1-y+x_i'}<-\frac{1}{1-y+x_i}$ and $\frac{1}{1+y-x_i'}<\frac{1}{1+y-x_i}$, we deduce that for both cases of $x_i\leq \textsc{NashFL}(\mathbf{x})$ and $x_i>\textsc{NashFL}(\mathbf{x})$, $\frac{dNash(y;x')}{dy}<0$ for $y\in (\textsc{NashFL}(\mathbf{x}),x_{k+1}')$\footnote{Note the length of the interval does not increase as we transition from $(\textsc{NashFL}(\mathbf{x}),x_{k+1})$ to $(\textsc{NashFL}(\mathbf{x}),x_{k+1}')$.}. Now if we move $y$ to interval $(x_{k+1}',x_{k+2}')$, the derivative remains negative as each individual fraction decreases, and the $\frac{1}{1+y-x_{k+1}}$ term changes to $-\frac{1}{1-y+x_{k+1}}$. Applying this argument to regions $(x_{k+2}',x_{k+3}'),\dots,(x_n',1]$, we see that $\frac{dNash(y;x')}{dy}<0$ for $y\in (\textsc{NashFL}(\mathbf{x}),1]\backslash \{x_{k+1}',\dots,x_{n}'\}$.

\textbf{Case 2:}

Suppose that after the shift there are $k+1$ agents left of $\textsc{NashFL}(\mathbf{x})$ (as $x_i>\textsc{NashFL}(\mathbf{x})$ and $x_i'\leq \textsc{NashFL}(\mathbf{x}))$. As a result, the fractional term corresponding to $x_i$ has changed from $\frac{1}{1+y-x_i}$ to $-\frac{1}{1-y+x_i'}$. We therefore have 
\begin{align*}
\bigg(\sum_{\substack{j=k+1 \\ j\neq i}}^n\frac{1}{1+y-x_j}-\frac{1}{1-y+x_i'}-\sum_{j=1}^k\frac{1}{1-y+x_j}\bigg)&<\bigg(-\sum_{j=1}^k \frac{1}{1-y+x_j} +\sum_{j=k+1}^n\frac{1}{1+y-x_j}\bigg)\\
&<0,
\end{align*}
implying that in this case, $\frac{dNash(y;\mathbf{x}')}{dy}<0$ for $y\in (\textsc{NashFL}(\mathbf{x}),x_{k+2}')$. Applying the same argument as in the previous case, we deduce that $\frac{dNash(y;x')}{dy}<0$ for $y\in (\textsc{NashFL}(\mathbf{x}),1]\backslash \{x_{k+2}',\dots,x_{n}'\}$.

By exhaustion of cases, we have shown that $Nash(y;\mathbf{x}')$ is strictly decreasing w.r.t $y$ in the interval $(\textsc{NashFL}(\mathbf{x}),1]$, implying that $\textsc{NashFL}(\mathbf{x}')\notin (\textsc{NashFL}(\mathbf{x}),1]$.
\end{proof}
\section{Proof of Lemma \ref{smaller}}
\begin{customlem}{4}
Suppose we have two different agent location profiles $\mathbf{x}=(x_1,\dots,x_n)$ and $\mathbf{x}'=(x_1',\dots,x_n')$. The following inequality holds:
\[|\textsc{NashFL}(\mathbf{x})-\textsc{NashFL}(\mathbf{x'})|\leq \max_{i\in N} |x_i-x_i'|.\]
\end{customlem}
\begin{proof}
%suppose without loss of generality that for some $j=\argmax_{i\in N} |x_i'-x_i|$, $x_j'-x_j>0$. In other words, the agent with the greatest change between $\mathbf{x}'$ and $\mathbf{x}$ has shifted to the right. 
Let $c=\max_{i\in N} |x_i-x_i'|$ and construct the agent location profiles $\mathbf{x}_{+c}=(x_1+c,\dots,x_n+c)$ and $\mathbf{x}_{-c}=(x_1-c,\dots,x_n-c)$. From Lemma \ref{invar}, we know that $\textsc{NashFL}(\mathbf{x}_{+c})=\textsc{NashFL}(\mathbf{x})+c$ and $\textsc{NashFL}(\mathbf{x}_{-c})=\textsc{NashFL}(\mathbf{x})-c$. Now for all $i\in N$, $x_i+c-x_i'\geq 0$, so by Theorem \ref{nonneg}, we have
\[\textsc{NashFL}(\mathbf{x}')\leq \textsc{NashFL}(\mathbf{x}_{+c}).\]
In other words, if we construct location profile $\mathbf{x}'$ from $\mathbf{x}_{+c}$ by shifting agents, none of the agents would be shifted to the right, so by Theorem \ref{nonneg}, $\textsc{NashFL}(\mathbf{x}')$ cannot be right of $\textsc{NashFL}(\mathbf{x}_{+c})$. By making a similar argument with $\mathbf{x}_{-c}$, we have
\[\textsc{NashFL}(\mathbf{x}')\geq \textsc{NashFL}(\mathbf{x}_{-c}).\]
Combining our inequalities, we have
\[\textsc{NashFL}(\mathbf{x})-c\leq \textsc{NashFL}(\mathbf{x}')\leq \textsc{NashFL}(\mathbf{x})+c\]
\[\implies -c\leq \textsc{NashFL}(\mathbf{x}')-\textsc{NashFL}(\mathbf{x})\leq c\]
\[\implies |\textsc{NashFL}(\mathbf{x})-\textsc{NashFL}(\mathbf{x'})|\leq \max_{i\in N} |x_i-x_i'|.\] 
\end{proof}
\section{Proof of Lemma \ref{yesyes}}
\begin{customlem}{5}
Let there be $k$ agents at location $x$ (where $0<x\leq 1$) and $n-k$ agents at location $0$. If $\frac{k}{n}\geq \frac{1}{2-x}$, then \textsc{NashFL} places the facility at $x$. If $\frac{k}{n}\leq \frac{1-x}{2-x}$, then \textsc{NashFL} places the facility at $0$. If neither of these inequalities hold, then \textsc{NashFL} places the facility at $x-1+\frac{2k-kx}{n}$.
\end{customlem}
\begin{proof}
Let $\mathbf{x}=(\underbrace{0,\dots,0}_{n-k},\underbrace{x,\dots,x}_{k})$. The Nash welfare is
\[Nash(y;\mathbf{x})=(1-y)^{n-k}(1-x+y)^k,\]
so its derivative with respect to facility location $y$ is
\begin{align*}
\frac{dNash(y;\mathbf{x})}{dy}&=k(1-y)^{n-k}(1-x+y)^{k-1}-(n-k)(1-y)^{n-k-1}(1-x+y)^k\\
&=(1-x+y)^{k-1}(1-y)^{n-k-1}(k(1-y)-(n-k)(1-x+y)).
\end{align*}
We assume that $y<1$, as $y=1$ leads to a Nash welfare of $0$. Now $(1-x+y)^{k-1}(1-y)^{n-k-1}>0$ for all $x\in (0,1]$ and $y\in [0,1]$. Therefore if $k(1-y)-(n-k)(1-x+y)\geq 0$ for all $y\in [0,x]$, then the Nash welfare is non-decreasing with respect to $y$, so it attains its maximum at $y=x$. We minimize this expression with respect to $y$ by substituting $y=x$, giving us the inequality $2k-kx-n\geq 0$. This is rearranged to form $\frac{k}{n}\geq \frac{1}{2-x}$.

Similarly, if $k(1-y)-(n-k)(1-x+y)\leq 0$ for all $y\in [0,x]$, then the Nash welfare is non-increasing with respect to $y$ and attains its maximum at $y=0$. We maximize this expression with respect to $y$ by substituting $y=0$ to attain the inequality $(2k-n)+x(n-k)\leq 0$. This is rearranged to form $\frac{k}{n}\leq \frac{1-x}{2-x}$.

If neither of these inequalities hold, then there exists $y$ in $(0,x)$ such that $\frac{dNash(y;\mathbf{x})}{dy}=0$. By rearranging terms in $k(1-y)-(n-k)(1-x+y)$, we see that this occurs when $y=x-1+\frac{2k-kx}{n}$. 
\end{proof}
\section{Proof of Theorem \ref{asdf}}
\begin{customthm}{6}
\textsc{Mid} has an approximation ratio for utilitarian social welfare of $2-\frac{2}{n}$.
\end{customthm}
\begin{proof}
Suppose that $x_1=0$ and $x_n=x$, where $x\in(0,1]$. From the proof of Theorem \ref{ez}, we have the following upper bound on the median solution utilitarian social welfare
\[
n-\sum^n_{i=1}|y_{MED}-x_i|\leq n-x.
\]
Since $x_i\in [0,x]$, we also have the following lower bound on the midpoint solution utilitarian social welfare
\[n-\sum^n_{i=1}|\frac{x}{2}-x_i|\geq n-\frac{nx}{2}.\]
We therefore have 
\[\frac{n-\sum^n_{i=1}|x_i-y_{MED}|}{n-\sum^n_{i=1}|x_i-y|}\leq \frac{2(n-x)}{2n-nx},\]
with equality when we have $n-1$ agents at $0$ and $1$ agent at $x$. This fraction also attains a maximum of $2-\frac{2}{n}$ when $x=1$. The theorem statement follows. 
\end{proof}

%\section{Proof of Lemma \ref{po}}
%\begin{customlem}{9}
%Any Pareto optimal facility location mechanism has an approximation ratio of at most $n-1$ for utilitarian social welfare.
%\end{customlem}
%\begin{proof}
%Let $y$ be the solution of the Pareto optimal facility location mechanism $f$, and $y_{MED}$ be the solution of the median mechanism. Also let $x_1=0$ and $x_n=x$, where $x\in (0,1]$. The approximation ratio for utilitarian social welfare is
%\[\max_{\mathbf{x}\in [0,1]^n} \frac{USW(\textsc{Med}(\mathbf{x}),\mathbf{x})}{USW(f(\mathbf{x}),\mathbf{x})}=\max_{\mathbf{x}\in [0,1]^n}\frac{n-\sum^n_{i=1}|x_i-y_{MED}|}{n-\sum^n_{i=1}|x_i-y|}.\]
%From the proof of Theorem \ref{ez}, we have $\sum^n_{i=1}|x_i-y_{MED}| \geq x$.
%We also have
%\begin{align*}
%\sum^n_{i=1}|x_i-y|&=\sum^{n-1}_{i=2}|x_i-y|+(y-x_1)+(x_n-y)\\
%&=x+\sum^{n-1}_{i=2}|x_i-y|\\
%&\leq (n-1)x.
%\end{align*}
%Therefore,
%\begin{align*}
%\frac{n-\sum^n_{i=1}|x_i-y_{MED}|}{n-\sum^n_{i=1}|x_i-y|}&\leq \frac{n-x}{n-(n-1)x}\\
%&\leq n-1.
%\end{align*}
%The last line is a result of the fraction attaining a maximum of $n-1$ at $x=1$. 
%\end{proof}
\section{Proof of Theorem~\ref{thm:midornear USW}}
\begin{customthm}{11}
\textsc{MidOrNearest} is a $(2-\frac{2}{n})-$approximation of the utilitarian social welfare.
\end{customthm}
\begin{proof}
For simplicity, let $f$ denote the \textsc{MidOrNearest} mechanism. We first show that any agent location profile $\mathbf{x}$ where all agent locations are strictly above (resp. below) $\frac{1}{2}$ can be modified to form a location profile $\mathbf{x}'$ where $x_1'\leq \frac{1}{2}\leq x_n'$, without decreasing the (utilitarian) welfare ratio. Let $\mathbf{x}$ be a location profile where all agents are strictly above $\frac{1}{2}$ (i.e. $x_1>\frac{1}{2}$), and let $\mathbf{x}'$ be such that $x_1'=\frac{1}{2}$ and $x_i'=x_i$ for $i\in \{2,\dots,n\}$. The optimal welfare decreases by $(x_1-\frac{1}{2})$, and the welfare corresponding to $f$ decreases by $(n-1)(x_1-\frac{1}{2})$ as the facility under $f$ moves to $x_1'=\frac{1}{2}$. We therefore have
\[\frac{USW(\textsc{Med}(\mathbf{x}'),\mathbf{x}')}{USW(f(\mathbf{x}'),\mathbf{x}')}=\frac{USW(\textsc{Med}(\mathbf{x}),\mathbf{x})-(x_1-\frac{1}{2})}{USW(f(\mathbf{x}),\mathbf{x})-(n-1)(x_1-\frac{1}{2})}\geq \frac{USW(\textsc{Med}(\mathbf{x}),\mathbf{x})}{USW(f(\mathbf{x}),\mathbf{x})}.\]
Due to symmetry, this expression also holds when $\mathbf{x}$ has all agents strictly below $\frac{1}{2}$.

Now let $\mathbf{x}'$ be an arbitrary location profile where $x_1'\leq \frac{1}{2}\leq x_n'$. Note that $f$ places the facility at $\frac{1}{2}$. We show that it can be transformed into a profile where all agents are located at an endpoint, without decreasing the welfare ratio. Let $x'_{med}:= x'_{\floor{\frac{n}{2}}}$ denote the median agent, and suppose without loss of generality that $x'_{med}\geq \frac{1}{2}$. Define $\mathbf{x}''$ as the location profile where $x''_i=1$ for all $i\in \{i:i>med\}$ and $x''_i=x'_i$ for all $i\in \{i:i\leq med\}$. Both the optimal welfare and the welfare under $f$ decrease by $\sum_{i>med}(1-x'_i)$. We therefore have
\[\frac{USW(\textsc{Med}(\mathbf{x}''),\mathbf{x}'')}{USW(f(\mathbf{x}''),\mathbf{x}'')}=\frac{USW(\textsc{Med}(\mathbf{x}'),\mathbf{x}')-\sum_{i>med}(1-x'_i)}{USW(f(\mathbf{x}'),\mathbf{x}')-\sum_{i>med}(1-x'_i)}\geq \frac{USW(\textsc{Med}(\mathbf{x}'),\mathbf{x}')}{USW(f(\mathbf{x}'),\mathbf{x}')}.\]

Now define $\mathbf{x}'''$ as the location profile where $x'''_{med}=1$ and $x'''_i=x''_i$ for all $i\neq med$. The optimal welfare decreases by $(1-x''_{med})$ if $n$ is even, and it does not change if $n$ is odd. Also, the welfare under $f$ decreases by $(1-x''_{med})$. We therefore have
\[\frac{USW(\textsc{Med}(\mathbf{x}'''),\mathbf{x}''')}{USW(f(\mathbf{x}'''),\mathbf{x}''')}=\frac{USW(\textsc{Med}(\mathbf{x}''),\mathbf{x}'')-(1-x''_{med})\mathbb{I}_{n \text{ even}}}{USW(f(\mathbf{x}''),\mathbf{x}'')-(1-x''_{med})}\geq \frac{USW(\textsc{Med}(\mathbf{x}''),\mathbf{x}'')}{USW(f(\mathbf{x}''),\mathbf{x}'')}.\]

Finally, define $\mathbf{x}''''$ as the location profile where $x''''_i=0$ for $i\in S:= \{i: x_i\leq \frac{1}{2}\}$ and $x''''_i=1$ for $i\in N\backslash S$. The optimal welfare decreases by $\sum_{i\in S}x'''_i$ from the agents in $S$ moving to $0$, and it also increases by $\sum_{i\in N\backslash S}(1-x'''_i)$ from the agents in $N\backslash S$ moving towards the median agent at $1$. The welfare corresponding to $f$ decreases by $\sum_{i\in S}x'''_i + \sum_{i\in N\backslash S}(1-x'''_i)$ from the agents moving away from the facility at $\frac{1}{2}$. We finally have
\begin{align*}
\frac{USW(\textsc{Med}(\mathbf{x}''''),\mathbf{x}'''')}{USW(f(\mathbf{x}''''),\mathbf{x}'''')}&=\frac{USW(\textsc{Med}(\mathbf{x}'''),\mathbf{x}''')-\sum_{i\in S}x'''_i+\sum_{i\in N\backslash S}(1-x'''_i)}{USW(f(\mathbf{x}'''),\mathbf{x}''')-\sum_{i\in S}x'''_i - \sum_{i\in N\backslash S}(1-x'''_i)}\\
&\geq \frac{USW(\textsc{Med}(\mathbf{x}'''),\mathbf{x}''')}{USW(f(\mathbf{x}'''),\mathbf{x}''')}\\
&\geq \frac{USW(\textsc{Med}(\mathbf{x}'),\mathbf{x}')}{USW(f(\mathbf{x}'),\mathbf{x}')}.
\end{align*}
We have shown that any location profile can be transformed into one where all agents are located at endpoints without decreasing the welfare ratio, meaning that the welfare ratio is maximized at such a profile. This implies that
\[\max_{\mathbf{x}\in [0,1]^n}\frac{USW(\textsc{Med}(\mathbf{x}),\mathbf{x})}{USW(f(\mathbf{x}),\mathbf{x})}=\max_{\mathbf{x}\in \{0,1\}^n}\frac{USW(\textsc{Med}(\mathbf{x}),\mathbf{x})}{USW(f(\mathbf{x}),\mathbf{x})}.\]
Now among the location profiles where agents are restricted to endpoints, the profile maximizing the welfare ratio has $n-1$ agents at one endpoint and $1$ agent at the other endpoint. This can be seen by taking an endpoint profile where there are at least $2$ agents at an endpoint, and moving agents from the endpoint with less agents to the endpoint with more agents, increasing the optimal welfare while keeping the welfare corresponding to $f$ constant. The approximation ratio of $2-\frac{2}{n}$ corresponds to the profile with $n-1$ agents at one endpoint and $1$ agent at the other endpoint, which has been proven to maximize the welfare ratio.
\end{proof}
\section{Proof of Theorem \ref{midoregal}}
\begin{customthm}{12}
\textsc{MidOrNearest} is a $\frac{3}{2}-$approximation of the egalitarian social welfare. No strategy-proof mechanism has a smaller approximation ratio.
\end{customthm}
\begin{proof}
There are three cases. In the first case, $x_1\leq \frac{1}{2}\leq x_n$ and the \textsc{MidOrNearest} mechanism locates the facility at $\frac{1}{2}$. The worst case for the approximation ratio occurs when $x_1=\frac{1}{2}$ and $x_n=1$. This gives an egalitarian social welfare of $\frac{1}{2}$ compared to an optimum of $\frac{3}{4}$. The approximation ratio is therefore $\frac{3}{2}$ at best. In the second case, $x_n\leq \frac{1}{2}$. The egalitarian social welfare is $1-(x_n-x_1)$ units. The optimal egalitarian social welfare is $1-\frac{(x_n-x_1)}{2}$ units. Define $f(z)=\frac{1-z}{1-2z}$ where $z=\frac{x_n-x_i}{2}$. For $z\in [0,\frac{1}{4}]$, $f(z)$ takes a maximum of $\frac{3}{2}$ at $z=\frac{1}{4}$, corresponding to $x_n=\frac{1}{2}$ and $x_1=0$. The approximation ratio is therefore $\frac{3}{2}$ at best. The third case, with $x_1>\frac{1}{2}$ is symmetric to the second case.

To show that no strategy-proof mechanism can have a smaller approximation ratio, suppose the opposite and that a mechanism exists with a smaller ratio. Consider two agents, $x_1=0$ and $x_2=1$. Suppose that the facility is located at $\frac{1}{2}+\epsilon$ for $\epsilon\geq 0$. The case where the facility is located at $\frac{1}{2}-\epsilon$ is dual. Suppose the second agent reports $x_2=\frac{1}{2}+\epsilon$. The optimal egalitarian social welfare is $\frac{3}{4}-\frac{\epsilon}{2}$. If the mechanism is to achieve an approximation ratio of less than $\frac{3}{2}$ then the minimum utility must be less than $\frac{1}{2}-\frac{\epsilon}{3}$. The facility must therefore be in $[0,\frac{1}{2}+\frac{\epsilon}{3})$. Therefore if there are two agents at $x_1=0$ and $x_2=\frac{1}{2}+\epsilon$, the second agent has an incentive to misreport their location as $x_2=1$. This contradicts the assumption that there is a strategy-proof mechanism with a smaller approximation ratio. 
\end{proof}
\end{document}